\documentclass{article}

\usepackage{amsfonts}
\usepackage{amsmath}
\usepackage{indentfirst,latexsym,bm}
\usepackage{amssymb}
\usepackage{amsmath}
\usepackage{graphics}
\usepackage{graphicx}
\usepackage{amsfonts}
\usepackage{amsmath}
\usepackage{indentfirst,latexsym}
\usepackage{amssymb}
\usepackage{amsmath}
\usepackage{graphics}
\usepackage{authblk}
\usepackage{color}
\usepackage{float}
\usepackage{latexsym}
\usepackage{booktabs}
\usepackage{multirow}
\usepackage{graphicx}
\usepackage{epsfig}
\usepackage{geometry}
\usepackage{amscd}
\usepackage{lscape}
\usepackage{float}
\usepackage{amsmath}
\usepackage{amssymb}
\usepackage{tabularx}
\usepackage{enumerate}
\usepackage{amsfonts}
\usepackage{mathrsfs}
\usepackage{eufrak}
\usepackage{graphicx}
\usepackage{dsfont}
\usepackage{slashbox}
\usepackage{epsfig}
\usepackage[ruled, lined, vlined, shortend, algosection]{algorithm2e}
\usepackage{longtable}
\usepackage{pdflscape}
\usepackage{lscape}
\usepackage{curves}
\usepackage{epic}
\usepackage{amsfonts}
\usepackage{amsmath}
\usepackage{indentfirst,latexsym,bm}
\usepackage{amssymb}
\usepackage{amsmath}
\usepackage{graphics}
\usepackage{authblk}
\usepackage{color}
\usepackage{float}
\usepackage{latexsym}
\usepackage{booktabs}
\usepackage{multirow}
\usepackage{graphicx}
\usepackage{epsfig}
\usepackage{geometry}
\usepackage{amscd}
\usepackage{lscape}

\numberwithin{equation}{section}

\newtheorem{theorem}{Theorem}[section]

\newtheorem{corollary}[theorem]{Corollary}

\newtheorem{definition}[theorem]{Definition}

\newtheorem{proposition}[theorem]{Proposition}
\newtheorem{remark}[theorem]{Remark}

\newenvironment{proof}[1][Proof]{\textbf{#1.} }{\ \rule{0.5em}{0.5em}}
\DeclareMathOperator*{\esssup}{ess\,sup}

\DeclareMathOperator*{\essinf}{ess\,inf}

\begin{document}

\title{\textbf{Pseudo Linear Pricing Rule for Utility Indifference Valuation}
\footnotetext{The authors thank the editor Martin Schweizer, an
associate editor, and two referees for their valuable suggestions,
and Damiano Brigo, Rama Cont, David Hobson, Monique Jeanblanc,
Lishang Jiang, Eva L\"utkebohmert, Andrea Macrina, Shige Peng,
Xingye Yue, and Thaleia Zariphopoulou for helpful discussions. The
authors thank participants in seminars at the University of
Freiburg, and King's College London, and at the Sino-French Summer
Institute in Stochastic Modeling and Applications (Beijing, June
2011), Stochastic Analysis: A UK-China Workshop (Loughborough, July
2011), and the Fourth International Conference: Mathematics in
Finance (South Africa, August 2011) where work in progress was
presented.}}

\author[$\dag$,$\sharp$]{\small{\textsc{Vicky Henderson}}}

\author[$\ddag$,$\sharp$]{\small{\textsc{Gechun Liang}}}

\affil[$\dag$]{\small{\textsc{Department of Statistics, University
of Warwick, Coventry, CV4 7AL, U.K.}}\\

\texttt{Vicky.Henderson@warwick.ac.uk}}

\affil[$\ddag$]{\small{\textsc{Department of Mathematics, King's College London, London, WC2R 2LS, U.K.}}\\

\texttt{gechun.liang@kcl.ac.uk}}

\affil[$\sharp$]{\small{\textsc{Oxford-Man Institute, University of
Oxford, Oxford, OX2 6ED, U.K.}}}
\date{}
\maketitle

\begin{abstract}
This paper considers exponential utility indifference pricing for a
multidimensional non-traded assets model, and provides two linear
approximations for the utility indifference price. The key tool is a
probabilistic representation for the utility indifference price by
the solution of a functional differential equation, which is termed
\emph{pseudo linear pricing rule}. We also provide an alternative
derivation of the quadratic BSDE representation for the utility
indifference price.\\


\end{abstract}

\newpage
\section{Introduction}

In this paper, we consider \emph{exponential} utility indifference
pricing in a multidimensional non-traded assets setting. Our
interest is in pricing and hedging derivatives written on assets
that are not traded. The market is incomplete as the risks arising
from having exposure to non-traded assets cannot be fully hedged.
There has been considerable research in the area of exponential
utility indifference valuation, but despite the interest in this
pricing and hedging approach, there have been relatively few
explicit formulas available. The well known {\it one dimensional
non-traded assets model} is an exception, and in a Markovian
framework with a derivative written on a single non-traded asset,
and partial hedging in a financial asset, Henderson and Hobson
\cite{Henderson2}, Henderson \cite{Henderson1}, and Musiela and
Zariphopoulou \cite{Musiela} use the Cole-Hopf transformation (or
\emph{distortion power}) to linearize the nonlinear partial
differential equations (\emph{PDEs} for short) for the value
function. Subsequent generalizations of the model from Tehranchi
\cite{Tehranchi}, Frei and Schweizer \cite{FSI} \cite{Schweizer}
have shown the exponential utility indifference price can still be
written in a closed-form expression similar to that known for the
Brownian setting, although the structure of the formula can be much
less explicit. On the other hand, Davis \cite{Davis} uses the
duality to derive an explicit formula for the optimal hedging
strategy, and Becherer \cite{Bec0} shows that the dual pricing
formula exists even in a general semimartingale setting.

We study exponential utility indifference valuation in a
multidimensional setting with the aim of developing a pricing
methodology. The main tool that we use to prescribe the pricing
dynamics is backward stochastic differential equation with quadratic
growth (\emph{quadratic BSDE} for short). It is well known in the
literature that the utility indifference price can be written as a
nonlinear expectation of the payoff under the original physical
measure, and the nonlinear expectation is often specified by a
quadratic BSDE. Several authors derive quadratic BSDE
representations of exponential utility indifference values in models
of varying generality - see Hu et al \cite{Hu}, Mania and Schweizer
\cite{MS}, Becherer \cite{Bec}, Morlais \cite{Morlais}, Frei and
Schweizer \cite{FSI} \cite{Schweizer}, Bielecki and Jeanblanc
\cite{Bielecki1}, and Ankirchner et al \cite{Imkeller1} among
others. Their derivations use the martingale optimality principle.

Our first contribution is to provide an alternative approach to
derive the quadratic BSDE representation of the utility indifference
price. The martingale optimality principle does not play any role in
our derivation. Instead, we consider the associated utility
maximization problems for utility indifference valuation from a
risk-sensitive control perspective. We first transform our utility
maximization problems into risk-sensitive control problems, and then
use the comparison principle for a family of quadratic BSDEs indexed
by the trading strategies to derive the pricing dynamics for the
utility indifference price. The details are presented in Theorem
\ref{Lemma for n nontraded assets}. We call such a quadratic BSDE
representation for the utility indifference price a \emph{nonlinear
pricing rule}.

With regards to the theory of quadratic BSDEs, the existence and
uniqueness of bounded solutions was first proved in a Brownian
setting by Kobylanski \cite{Kobylanski}, and was extended to
unbounded solutions and convex driver by Briand and Hu
\cite{MR2257138,MR2391164}, and Delbaen et al
\cite{Delbaen0,Delbaen}. The corresponding semimartingale case for
bounded solutions may be found in Morlais \cite{Morlais} and
Tevzadze \cite{Tevzadze}, where in the former, the main theorems of
\cite{Kobylanski} and \cite{Hu} were extended, and in the latter, a
fixed point argument with \emph{BMO} martingale techniques was
employed. See also Mocha and Westray \cite{Mocha} for the extension
to unbounded solutions. In addition, Ankirchner et al \cite{AIR} and
Imkeller et al \cite{IRR} consider the differentiability of
quadratic BSDEs, and Frei et al \cite{FMS} give convex bounds for
the solutions. More recently, Barrieu and El Karoui \cite{Barrieu}
introduce a notion of quadratic semimartingale to study the
stability of solutions, while Briand and Elie \cite{Briand} find a
simplified approach which was used to study the corresponding
delayed equations. Finally, quadratic BSDEs with jumps were studied
by Becherer \cite{Bec} for bounded solutions, and by El Karoui et al
\cite{Matoussi} for unbounded solutions.

Our main contribution is the provision of a new pricing formula for
the utility indifference price, which we call a \emph{pseudo linear
pricing rule}. In Theorem \ref{lemmapesudo}, we represent the
utility indifference price as a linear expectation of the payoff
plus a pricing premium, where the latter is represented by the
solution of a functional differential equation. Such an idea is
motivated by Liang et al \cite{LLQ2009}, where they transform BSDEs
to functional differential equations, and solve them on a general
filtered probability space. One of the advantages of such a
representation is that we only need to solve a functional
differential equation in order to calculate the utility indifference
price. Moreover, the functional differential equation runs forwards,
avoiding the conflicting nature between the backward equation and
the underlying forward equation.

To apply such a \emph{pseudo linear pricing rule}, we provide two
linear approximations for the utility indifference price. In
contrast to \cite{LLQ2009}, where the driver is Lipschitz
continuous, the driver of the functional differential equation
considered in this paper is quadratic. Nevertheless, we can employ
Picard iteration to approximate its solution. The first linear
approximation is based on perturbations of the functional
differential equation, the idea of which is motivated by Proposition
2 of Tevzadze \cite{Tevzadze}. The second linear approximation is
based on a nonlinear version of Girsanov's transformation in order
to vanish the driver. Such an idea has appeared in the BSDE
literature, for example, as Proposition 11 in Mania and Schweizer
\cite{MS} and measure solutions of BSDEs in Ankirchner et al
\cite{Imkeller0},
where they model the terminal data (the payoff) as a general random
variable. In contrast, as we specify the dynamics of the underlying
and the payoff structure, a coupled forward backward stochastic
differential equation (\emph{FBSDE} for short) appears naturally.

The paper is organized as follows: In Section \ref{sec-model}, we
present our multidimensional non-traded assets model, and present
the nonlinear pricing rule, i.e. the quadratic BSDE representation
for the utility indifference price. In Section \ref{sec-pseudo}, we
present our pseudo linear pricing rule, i.e. the functional
differential equation representation for the utility indifference
price, and present two linear approximations for the utility
indifference price based on such a representation formula.


\section{Quadratic BSDEs and Nonlinear Pricing Rule} \label{sec-model}

Let ${\cal{W}}=(W^1,\cdots,W^d)$ be a $d$-dimensional Brownian
motion on some filtered probability space $(\Omega,\mathcal{F},
\{{\mathcal{F}}_t\} ,\mathbf{P})$ satisfying the \emph{usual
conditions}, where $\mathcal{F}_t$ is the augmented $\sigma$-algebra
generated by $ ({\cal{W}}_u:0 \leq u \leq t)$. The market consists
of a traded financial index $P$, whose discounted price process is
given by
\begin{equation}\label{P_equ}
P_t=P_0+\int_0^tP_s(\mu^P_sds+\langle \sigma_s^P,
d\mathcal{W}_s\rangle),
\end{equation}
and a set of observable but non-traded assets
$\mathcal{S}=(S^1,\cdots,S^d)$, whose discounted price processes are
given by
\begin{equation}\label{dynamics of nontraded assets }
S_t^i=S_0^i+\int_0^tS_s^i(\mu^i_sds+\langle \sigma^i_s,
d\mathcal{W}_s\rangle)
\end{equation}
for $i=1,\cdots,d$. $\langle\cdot,\cdot\rangle$ denotes the inner
product in $\mathbb{R}^d$ with its Euclidean norm $||\cdot||$. We
have $\mu^P_t, \mu^i_t \in \mathbb{R} $,
$\sigma^P_t=(\sigma^{P1}_t,\cdots,\sigma^{Pd}_t) \in \mathbb{R}^d$
and $\sigma^i_t=(\sigma^{i1}_t,\cdots,\sigma^{id}_t) \in
\mathbb{R}^d$. There is also a risk-free bond or bank account with
discounted price $B_t=1$ for $t\geq 0$.

Our interest will be in pricing and hedging (path-dependent)
contingent claims written on the non-traded assets $\mathcal{S}$.
Specifically, we are concerned with contracts with the payoff at
maturity $T$ of $g({\cal{S}}_{\cdot})$, which may depend on the
whole path of $\mathcal{S}$. We impose the following assumptions,
which will hold throughout:
\begin{itemize}
\item \textbf{Assumption (A1)}: All the coefficients
$\mu_t^{i}(\omega)$, $\sigma_t^{i}(\omega)$, $\mu_t^P(\omega)$ and
$\sigma_t^{P}(\omega)$ are $\mathcal{F}_t$-predictable and uniformly
bounded in $(t,\omega)$.

\item \textbf{Assumption (A2)}: The volatility
for the financial index $P$ is uniformly elliptic:
$||\sigma_t^P(\omega)||\geq\epsilon>0$ for all $(t,\omega)$.

\item \textbf{Assumption (A3)}: The payoff $g(\mathcal{S}_{\cdot})$, as a functional of the stochastic process
$\mathcal{S}$, is positive and bounded.

\end{itemize}

Our approach is to consider utility indifference valuation for such
contingent claims. For a general overview of utility indifference
valuation, we refer to the monograph edited by Carmona
\cite{MR2547456}, and especially the survey article by Henderson and
Hobson \cite{Henderson3} therein. For this we need to consider the
optimization problem for the investor both with and without the
option. The investor has initial wealth $X_t\in\mathcal{F}_t$ at any
starting time $t\in[0,T]$, and is able to trade the financial index
with price $P_t$ (and riskless bond with price $1$). This will
enable the investor to partially hedge the risks she is exposed to
via her position in the claim.

The holder of the option has an exponential utility function with
respect to her terminal wealth:
$$U_T(x)=-e^{-\gamma x}\ \ \ \text{for}\ \gamma> 0\ \text{and}\ x\in\mathbb{R}.$$

The investor holds $\lambda$ units of the claim, whose price at time
$t\in[0,T]$ is denoted as $\mathfrak{C}_t^{\lambda}$ and is to be
determined,
and invests her remaining wealth $X_t-\mathfrak{C}_t^{\lambda}$ in
the financial index $P$. The investor will follow an admissible
trading strategy
\begin{align*}
\pi\in\mathcal{A}_{ad}[0,T]=&\left\{\pi:\pi\ \text{is}\
\mathcal{F}_t\text{-predictable,}\
\sup_{\tau}\left\| E^{\mathbf{P}}\left[\left.\int_{\tau}^T|\pi_t|^2dt\right|\mathcal{F}_{\tau}\right]\right\|_{\infty}<\infty\right.\\
&\left.\text{for any}\ \mathcal{F}_t\text{-stopping time}\
\tau\in[0,T],\ \text{and\ moreover,}\right.\\[+0.2cm]
&\left.\epsilon\leq
E^{\mathbf{P}}\left[\left.e^{-\gamma\int_t^{T}\frac{\pi_s}{P_s}dP_s}\right|\mathcal{F}_t\right]\leq
K \ \text{for\ a.e}\ (t,\omega)\in[0,T]\times\Omega\right\},
\end{align*}
for some constants $K\geq \epsilon>0$, which results in the
following wealth equation: For $0\leq t\leq s\leq T$,
\begin{align}\label{X_equ}
X_s^{X_t-\mathfrak{C}_t^{\lambda}}(\pi)&=X_t-
\mathfrak{C}_t^{\lambda}+\int_t^{s}\frac{\pi_u}{P_u}dP_u\nonumber\\
&=X_t-
\mathfrak{C}_t^{\lambda}+\int_t^s\pi_u\left(\mu^P_udu+\langle\sigma^P_u,d\mathcal{W}_u\rangle\right).
\end{align}

\begin{remark}
The integrability conditions on the trading strategies $\pi$ are
slightly different from those required in Definition 1 of Hu et al
\cite{Hu}. They assume
$E^{\mathbf{P}}[\int_0^{T}|\pi_t|^2dt]<\infty$ to guarantee the
no-arbitrage condition, and that the utility of the gain process
$-e^{-\gamma\int_0^{\cdot}\frac{\pi_s}{P_s}dP_s}$ is in Doob's class
$\mathcal{D}$ in order to apply the martingale optimality principle.
Our first integrability condition is nothing but the \emph{BMO}
martingale property of $\int_0^{\cdot}\pi_sd\mathcal{W}_s$. The
second condition is about the integrability of the utility of the
remaining gain process
$-e^{-\gamma\int_{\cdot}^{T}\frac{\pi_s}{P_s}dP_s}$.
Both of the integrability conditions are needed in order to derive
the quadratic BSDE representation for the utility indifference price
in the following Theorem \ref{Lemma for n nontraded assets}.
However, they are not restrictive if we only price and hedge
contingent claims with bounded payoff, as the corresponding optimal
trading strategy satisfies these conditions anyway, and coincides
with the optimal trading strategy obtained in \cite{Hu}.
\end{remark}

We recall a continuous martingale $M$ with
$E^{\mathbf{P}}{[M,M]_T}<\infty$ is called a $\mathbf{P}$-\emph{BMO}
martingale if
$$\sup_{\tau}\left\|E^{\mathbf{P}}[|M_T-M_{\tau}|^2|\mathcal{F}_{\tau}]\right\|_{\infty}<\infty$$
for any $\mathcal{F}_t$-stopping time $\tau\in[0,T]$. By Theorem 2.3
of Kazamaki \cite{Kazamaki}, if $M$ is a $\mathbf{P}$-\emph{BMO}
martingale, its Dol\'eans-Dade exponential $\mathcal{E}(M)$ is in
Doob's class $\mathcal{D}$, and therefore uniformly integrable.
Another useful property that will be used later is the following
version of the John-Nirenburg inequality
\begin{equation}\label{JN_inequality}
\sup_{\tau}\left\|E^{\mathbf{P}}[|M_T-M_{\tau}|^2|\mathcal{F}_{\tau}]\right\|_{\infty}
\leq
K_1\sup_{\tau}\left\|E^{\mathbf{P}}[|M_T-M_{\tau}||\mathcal{F}_{\tau}]\right\|^2_{\infty}
\end{equation}
for some constant $K_1>0$ (see Corollary 2.1 of \cite{Kazamaki}).

The investor will optimize over the admissible trading strategies to
choose an optimal ${\pi}^{*,\lambda}$ by maximizing her expected
terminal utility
\begin{equation}\label{optm1}
\esssup_{\pi\in\mathcal{A}_{ad}[t,T]}E^{\mathbf{P}}\left[\left.-e^{-\gamma\left(X_T^{X_t-
\mathfrak{C}_t^{\lambda}}(\pi)+\lambda  g({\cal{S}}_{\cdot})
\right)}\right|\mathcal{F}_t\right].
\end{equation}

To define the utility indifference price for the option, we also
need to consider the optimization problem for the investor without
the option. This involves the investor investing only in the
financial index itself. Her wealth equation is the same as
(\ref{X_equ}) but starts from $X_t$, and she will choose an optimal
$\pi^{*,0}$ by maximizing
\begin{equation}\label{optm2}
\esssup_{\pi\in\mathcal{A}_{ad}[t,T]}E^{\mathbf{P}}\left[\left.-e^{-\gamma
X_T^{X_t}(\pi)}\right|\mathcal{F}_t\right].
\end{equation}
We note that (\ref{optm2}) is a special case of (\ref{optm1}) with
$\lambda=0$.

\begin{definition}\label{definition1} (Utility indifference valuation and hedging)

The utility indifference price $\mathfrak{C}_t^{\lambda}$ of
$\lambda$ units of the derivative with payoff $g({\cal{S}}_{\cdot})$
is defined by the solution to
\begin{equation*}
\esssup_{\pi\in\mathcal{A}_{ad}[t,T]}E^{\mathbf{P}}\left[\left.-e^{-\gamma\left(X_T^{X_t-
\mathfrak{C}_t^{\lambda}}(\pi)+\lambda  g({\cal{S}}_{\cdot})
\right)}\right|\mathcal{F}_t\right]=
\esssup_{\pi\in\mathcal{A}_{ad}[t,T]}E^{\mathbf{P}}\left[\left.-e^{-\gamma
X_T^{X_t}(\pi)}\right|\mathcal{F}_t\right].
\end{equation*}
The hedging strategy for $\lambda$ units of the derivative is
defined by the difference in the two optimal trading strategies
$\pi^{*,\lambda}-\pi^{*,0}$.
\end{definition}


The main result of this section is to show that the price of the
option and the corresponding hedging strategy can be represented by
the solution of a quadratic BSDE.

\begin{theorem} (Nonlinear pricing rule)\label{Lemma for n nontraded assets}

Suppose that Assumptions (A1) (A2), and (A3) are satisfied. If
$(Y^{\lambda},\cal{Z}^{\lambda})$ is the unique solution of the
following quadratic BSDE
\begin{align}\label{BSDE for n nontraded assets}
Y_t^{\lambda}=&\ \left(\lambda
g({\cal{S}}_{\cdot})+\int_0^T\theta_sds\right)+\int_t^TF_s({\cal{Z}}^{\lambda}_s)ds
-\int_t^T\langle{\cal{Z}}^{\lambda}_s, d{\cal{W}}_s\rangle,
\end{align}
with $\theta_s=\frac{|\mu_s^P|^2}{2\gamma||\sigma_s^P||^2}$, and the
driver $F_s(z)$ given by
$$F_s(z)=-\frac{\gamma}{2}||z||^2+
\frac{\gamma}{2||\sigma_s^P||^2}\left|\langle\sigma_s^P,z\rangle
-\frac{\mu_s^P}{\gamma}\right|^2-\theta_s$$ for $z\in\mathbb{R}^d$,
then the utility indifference price $\mathfrak{C}_t^{\lambda}$ is
represented by the solution of the quadratic BSDE (\ref{BSDE for n
nontraded assets})
\begin{equation}\label{representationformula1}
\mathfrak{C}_t^{\lambda}=Y_t^{\lambda}-Y_t^{0},
\end{equation}
and the hedging strategy for $\lambda$ units of the option is given
by
$$-\frac{\langle\sigma_t^P,\mathcal{Z}_t^{\lambda}-\mathcal{Z}_t^0\rangle}{||\sigma_t^P||^2}.$$
\end{theorem}

\begin{remark}
It is known that the above type of quadratic BSDE (\ref{BSDE for n
nontraded assets}) can be derived by the martingale optimality
principle - see, for example, Theorem 7 of Hu et al \cite{Hu} and
Section 3 of Ankirchner et al \cite{Imkeller1} in a Brownian motion
setting, and Theorem 13 of Mania and Schweizer \cite{MS} and Section
2.1 of Morlais \cite{Morlais} in a general semimartingale setting.
In the following, we provide a new proof of Theorem \ref{Lemma for n
nontraded assets}, where the martingale optimality principle does
not play any role. Instead, we consider the problem from a
risk-sensitive control perspective, and use the comparison principle
for a family of quadratic BSDEs (\ref{QBSDEwithconrol2}) indexed by
the admissible trading strategy $\pi$ to derive the quadratic BSDE
representation. Although this technique is known in the literature
(see Sections 19-21 of Quenez \cite{Quenez} and Section 3 of El
Karoui et al \cite{ElKaroui19973}), we apply it for the first time
in the context of quadratic BSDEs with unbounded random
coefficients. On the other hand, treating utility indifference
valuation from a risk-sensitive control viewpoint may also lead to
new perspectives in utility maximization problems.
\end{remark}

\begin{proof} We consider the utility maximization
problem (\ref{optm1}). By using (\ref{X_equ}) in (\ref{optm1}), we
have
\begin{align*}
&\esssup_{\pi\in\mathcal{A}_{ad}[t,T]}E^{\mathbf{P}}
\left[\left.-e^{-\gamma\left(X_t-\mathfrak{C}_t^{\lambda}
+\int_t^T\frac{\pi_s}{P_s}dP_s +\lambda g({\cal{S}}_{\cdot})
\right)}\right|\mathcal{F}_t\right]\\[+0.1cm]
=&-e^{-\gamma(X_t-\mathfrak{C}_t^{\lambda})}\essinf_{\pi\in\mathcal{A}_{ad}[t,T]}
E^{\mathbf{P}}
\left[\left.e^{-\gamma\left(\int_t^T\frac{\pi_s}{P_s}dP_s +\lambda
g({\cal{S}}_{\cdot})
\right)}\right|\mathcal{F}_t\right]\\[+0.1cm]
=&-e^{-\gamma(X_t-\mathfrak{C}_t^{\lambda})}\essinf_{\pi\in\mathcal{A}_{ad}[t,T]}
E^{\mathbf{P}}
\left[\left.e^{-\gamma\left(\int_t^T\frac{\pi_s}{P_s}dP_s-\theta_sds\right)}
e^{-\gamma\left(\lambda
g({\cal{S}}_{\cdot})+\int_0^T\theta_sds\right)}
\right|\mathcal{F}_t\right]e^{\gamma\int_0^{t}\theta_sds}\\[+0.1cm]
=&-e^{-\gamma(X_t-\mathfrak{C}_t^{\lambda})}
\exp\left\{-\gamma\esssup_{\pi\in\mathcal{A}_{ad}[t,T]}Y_t^{\lambda}(\pi)\right\}e^{\gamma\int_0^{t}\theta_sds},
\end{align*}
where $Y_t^{\lambda}(\pi)$ denotes the risk-sensitive control
criterion
$$Y_t^{\lambda}(\pi)=-\frac{1}{\gamma}
\ln E^{\mathbf{P}} \left[\left.
e^{-\gamma\left(\int_t^T\pi_s\left(\mu^P_sds+\langle\sigma^P_s,d\mathcal{W}_s\rangle\right)-\theta_sds\right)}
e^{-\gamma\left(\lambda
g({\cal{S}}_{\cdot})+\int_0^T\theta_sds\right)}
\right|\mathcal{F}_t\right].$$ By Assumptions (A1)-(A3),
$\left|\lambda g(\cdot)+\int_0^{t}\theta_sds\right|\leq K_2$ for
some constant $K_2>1$. Moreover, by the conditions on the admissible
trading strategy $\pi$, we know that $Y_t^{\lambda}(\pi)$ is bounded
for $a.e.\ (t,\omega)\in[0,T]\times\Omega$:
$$-\frac{1}{\gamma}\ln K-K_2\leq Y_t(\pi)\leq-\frac{1}{\gamma}\ln \epsilon+K_2.$$

We further introduce the risk-sensitive control problem
$$\hat{Y}_t^{\lambda}=\esssup_{\pi\in\mathcal{A}_{ad}[t,T]}Y_t^{\lambda}(\pi).$$
In the following, we characterize both $Y^{\lambda}(\pi)$ and
$\hat{Y}^{\lambda}$ by the solutions of quadratic BSDEs.

First, we consider the risk-sensitive control criterion
$Y^{\lambda}(\pi)$ under a different probability measure. For any
given trading strategy $\pi\in\mathcal{A}_{ad}[0,T]$, we define a
$\mathbf{P}$-\emph{BMO} martingale
\begin{equation*}
N_t(\pi)=-\int_0^{t}\gamma\pi_s\langle\sigma^P_s,d\mathcal{W}_s\rangle,\
\ \ \text{for}\ t\in[0,T].
\end{equation*}
Indeed, for any $\mathcal{F}_t$-stopping time $\tau\in[0,T]$, by
using the conditions on the admissible trading strategy $\pi$ and
Assumption (A1) we have
\begin{align*}
\sup_{\tau}\left\|E^{\mathbf{P}}[|N_T(\pi)-N_{\tau}(\pi)|^2|\mathcal{F}_{\tau}]\right\|_{\infty}
=&\ \sup_{\tau}\left\|E^{\mathbf{P}}\left[\left.\int_{\tau}^T\gamma^2||\sigma_s^P||^2|\pi_s|^2ds\right|\mathcal{F}_{\tau}\right]\right\|_{\infty}\\
\leq&\
K\sup_{\tau}\left\|E^{\mathbf{P}}\left[\left.\int_{\tau}^T|\pi_s|^2ds\right|\mathcal{F}_{\tau}\right]\right\|_{\infty}<\infty
\end{align*}
for some constant $K>0$.  Hence, the Dol\'eans-Dade exponential
$\mathcal{E}(N(\pi))$ is uniformly integrable, and we change the
probability measure by defining
$$\frac{d\mathbf{P}^{\pi}}{d\mathbf{P}}=\mathcal{E}(N(\pi))=
\mathcal{E}(-\int_0^{\cdot}\gamma\pi_s\langle\sigma^P_s,d\mathcal{W}_s\rangle).$$
The risk-sensitive control criterion under $\mathbf{P}^{\pi}$
becomes
\begin{align*}
Y_t^{\lambda}(\pi)&= -\frac{1}{\gamma}\ln E^{\mathbf{P}^{\pi}}
\left[\left.e^{-\int_t^T(\gamma\mu_s^P\pi_s-\frac{1}{2}\gamma^2||\sigma_s^P||^2|\pi_s|^2-\gamma\theta_s)ds}
e^{-\gamma\left(\lambda g(\mathcal{S}_{\cdot})+\int_0^{T}\theta_sds
\right)}\right|\mathcal{F}_t\right]\\
&= -\frac{1}{\gamma}\ln E^{\mathbf{P}^{\pi}}
\left[\left.e^{-\gamma\int_t^TF_s(\pi)ds} e^{-\gamma\left(\lambda
g(\mathcal{S}_{\cdot})+\int_0^{T}\theta_sds
\right)}\right|\mathcal{F}_t\right],
\end{align*}
where we denote
$$F_s(\pi)=\mu_s^P\pi_s-\frac{\gamma}{2}||\sigma_s^P||^2|\pi_s|^2-\theta_s.$$
Note that $F_s(\pi)$ only depends on $\pi_s$, not on all of $\pi$.


Next, we characterize $Y^{\lambda}(\pi)$ by the solution of a
quadratic BSDE with unbounded random coefficients, whose existence
is proved by directly showing that $Y^{\lambda}(\pi)$ indeed
satisfies this BSDE. Indeed, note that for $t\in[0,T]$,
$$\bar{Y}_t^{\lambda}(\pi)=e^{-\gamma\left(Y_t^{\lambda}(\pi)+\int_0^{t}F_s(\pi)ds\right)}$$
is a uniformly integrable martingale under $\mathbf{P}^{\pi}$, since
$$e^{-\gamma\left(Y_t^{\lambda}(\pi)+\int_0^{t}F_s(\pi)ds\right)}=
E^{\mathbf{P}^{\pi}} \left[\left.e^{-\gamma\int_0^TF_s(\pi)ds}
e^{-\gamma\left(\lambda g(\mathcal{S}_{\cdot})+\int_0^{T}\theta_sds
\right)}\right|\mathcal{F}_t\right].
$$
By the martingale representation theorem, there exists an
$\mathcal{F}_t$-predictable process
$\bar{\mathcal{Z}}^{\lambda}(\pi)$ such that
\begin{equation}\label{linearBSDE}
\bar{Y}_t^{\lambda}(\pi) =e^{-\gamma\left(\lambda
g(\mathcal{S}_{\cdot})+\int_0^{T}(\theta_s+F_s(\pi))ds\right)}-
\int_t^{T}\langle\bar{\mathcal{Z}}_s^{\lambda}(\pi),d\mathcal{W}_s(\pi)\rangle,
\end{equation}
where $\mathcal{W}(\pi)=\mathcal{W}-[\mathcal{W},N(\pi)]$ is
Brownian motion under $\mathbf{P}^{\pi}$ by Girsanov's
transformation.

For any $t\in[0,T]$, if we define
$\mathcal{Z}_t^{\lambda}(\pi)=-\frac{1}{\gamma}\bar{\mathcal{Z}}^{\lambda}_t(\pi)/\bar{Y}_t^{\lambda}(\pi)$,
and apply It\^o's formula to
$Y_t^{\lambda}(\pi)=-\frac{1}{\gamma}\ln\bar{Y}_t^{\lambda}(\pi)-\int_0^{t}F_s(\pi)ds$,
then it is easy to verify that $(Y^{\lambda}(\pi),Z^{\lambda}(\pi))$
is a solution to the following quadratic BSDE
\begin{align}\label{QBSDEwithconrol}
Y_t^{\lambda}(\pi)=&\ \left(\lambda
g(\mathcal{S}_{\cdot})+\int_0^T\theta_sds\right)\nonumber\\
&\
+\int_t^{T}\left(F_s(\pi)-\frac{\gamma}{2}||\mathcal{Z}_s^{\lambda}(\pi)||^2\right)ds
-\int_t^{T}\langle\mathcal{Z}_s^{\lambda}(\pi),d\mathcal{W}_s(\pi)\rangle.
\end{align}
Equivalently under the original probability measure $\mathbf{P}$, we
write
\begin{align}\label{QBSDEwithconrol2}
Y_t^{\lambda}(\pi)=&\ \left(\lambda
g(\mathcal{S}_{\cdot})+\int_0^T\theta_sds\right)\nonumber\\
&\
+\int_t^{T}\left(F_s(\pi)-\gamma\langle\sigma_s^P,\mathcal{Z}_s^{\lambda}(\pi)\rangle\pi_s
-\frac{\gamma}{2}||\mathcal{Z}_s^{\lambda}(\pi)||^2\right)ds-
\int_t^{T}\langle\mathcal{Z}^{\lambda}_s(\pi),d\mathcal{W}_s\rangle.
\end{align}

We notice that BSDE (\ref{QBSDEwithconrol2}) has quadratic growth in
$z$ with unbounded random coefficients, which satisfy the \emph{BMO}
condition in Theorem 8 of Mania and Schweizer \cite{MS}. Since the
solution $Y^{\lambda}(\pi)$ is bounded, Theorem 8 of \cite{MS} then
implies that a comparison theorem holds for
(\ref{QBSDEwithconrol2}). Let $\pi^1,\pi^2\in\mathcal{A}_{ad}[0,T]$
such that
$$F_s(\pi^1)-\gamma\langle\sigma_s^P,z\rangle\pi_s^1
\geq F_s(\pi^2)-\gamma\langle\sigma_s^P,z\rangle\pi_s^2$$ for
$z\in\mathbb{R}^d$. Then $Y_t^{\lambda}(\pi^1)\geq
Y_t^{\lambda}(\pi^2)$ for $a.e.$ $(t,\omega)$. This can be proved
either by changing probability measure as in \cite{MS}, or by an
exponential change of variables.

As a byproduct, we also obtain that the quadratic BSDE
(\ref{QBSDEwithconrol2}) admits a unique solution
$(Y^{\lambda}(\pi),\mathcal{Z}^{\lambda}(\pi))$, where
$Y^{\lambda}(\pi)$ is a bounded special semimartingale with
$\mathcal{Z}^{\lambda}(\pi)$ as its corresponding martingale
representation.

Thirdly, we prove that the solution of our risk-sensitive control
problem is given by
\begin{equation}\label{first_claim}
\hat{Y}_t^{\lambda}=Y_t^{\lambda},
\end{equation}
and the optimal trading strategy is given by
\begin{equation}\label{second_claim}
\pi_s^{*,\lambda}=-\frac{\langle\sigma_s^P,\mathcal{Z}^{\lambda}_s\rangle}{||\sigma_s^P||^2}+\frac{\mu_s^P}{\gamma||\sigma_s^P||^2}
\end{equation} for $s\in[t,T]$, where
$(Y^{\lambda},\mathcal{Z}^{\lambda})$ solves the quadratic BSDE
(2.5), whose existence and uniqueness is guaranteed by Theorems 2.3
and 2.6 of Kobylanski [28] or Theorem 1 of Tevzadze [36]. Indeed,
the driver of (2.5) satisfies $F_t(0)=0$, and is smooth in $z$ with
\begin{align*}
\nabla_{z}F_t(z)&=-\gamma
z+\frac{\gamma}{||\sigma_t^{P}||^2}\left(\langle\sigma_t^P,z\rangle
-\frac{\mu_t^P}{\gamma}\right)\sigma_t^{P},\\
\nabla_{zz}F_t(z)&=-\gamma\mathbf{1}+\frac{\gamma}{||\sigma_t^{P}||^2}(\sigma_t^{P})^{T}\sigma_t^P,
\end{align*}
where the superscript $T$ denotes the matrix transposition. Hence,
by Assumptions (A1)-(A3),
\begin{align}\label{const0}
\begin{split}
||\nabla_{z}F_t(z)||\leq K_2(1+||z||),\\
||z||^2/K_2\leq z\nabla_{zz}F_t(z)z^T\leq K_2||z||^2,
\end{split}
\end{align}
and the terminal data satisfies
\begin{equation}\label{const1}
\left|\lambda g(\cdot)+\int_0^{T}\theta_sds\right|\leq K_2
\end{equation}
for some constant $K_2\geq 1$. Therefore, there exists a unique
solution $(Y^{\lambda},\mathcal{Z}^{\lambda})$ to BSDE (2.5), where
$Y^{\lambda}$ is a bounded special semimartingale with
$\mathcal{Z}^{\lambda}$ as its martingale representation. Moreover,
the martingale part
$\int_0^{\cdot}\langle\mathcal{Z}^{\lambda},d\mathcal{W}\rangle$ is
a $\mathbf{P}$-\emph{BMO} martingale.

Now we proceed to prove (\ref{first_claim}) and
(\ref{second_claim}). Notice that for any
$\pi\in\mathcal{A}_{ad}[t,T]$,
\begin{equation*}
F_s(\pi)-\gamma\langle\sigma_s^p,\mathcal{Z}^{\lambda}_s\rangle\pi_s
-\frac{\gamma}{2}||\mathcal{Z}^{\lambda}_s||^2=-\frac{\gamma}{2}||\sigma_s^P||^2|\pi_s-\pi_s^{*,\lambda}|^2
+F_s(\mathcal{Z}_s^{\lambda})\leq F_s(\mathcal{Z}_s^{\lambda}),
\end{equation*}
and for $\pi=\pi^{*,\lambda}$,
\begin{equation*}
F_s(\pi^{*,\lambda})-\gamma\langle\sigma_s^p,\mathcal{Z}^{\lambda}_s\rangle\pi_s^{*,\lambda}
-\frac{\gamma}{2}||\mathcal{Z}^{\lambda}_s||^2=
F_s(\mathcal{Z}_s^{\lambda}).
\end{equation*}
If $\pi^{*,\lambda}\in\mathcal{A}_{ad}[t,T]$, then by applying the
comparison theorem for the quadratic BSDE (2.7), we obtain that
$Y_t^{\lambda}(\pi)\leq Y_t^{\lambda}$ for any
$\pi\in\mathcal{A}_{ad}[t,T]$, and
$Y_t^{\lambda}(\pi^{*,\lambda})=Y_t^{\lambda}$.

Since
$\hat{Y}^{\lambda}_t=\sup_{\pi\in\mathcal{A}_{ad}[t,T]}Y^{\lambda}_t(\pi)$,
we get $\hat{Y}^{\lambda}_t= Y_t^{\lambda}$ and $\pi^{*,\lambda}$
achieves the maximum. We are left to verify that
$\pi^{*,\lambda}\in\mathcal{A}_{ad}[t,T]$. For this, we only need to
note that
$\int_0^{\cdot}\langle\mathcal{Z}^{\lambda}_s,d\mathcal{W}_s\rangle$
is a $\mathbf{P}$-\emph{BMO} martingale, and
$$Y_t^{\lambda}=Y_t^{\lambda}(\pi^{*,\lambda})= -\frac{1}{\gamma}
\ln E^{\mathbf{P}} \left[\left.
e^{-\gamma\int_t^T\frac{\pi^{*,\lambda}_s}{P_s}dP_s}
e^{-\gamma\lambda g({\cal{S}}_{\cdot})}
\right|\mathcal{F}_t\right]$$ is bounded for $a.e.$
$(t,\omega)\in[0,T]\times\Omega$, so is
$E^{\mathbf{P}}[e^{-\gamma\int_{t}^{T}\frac{\pi_s^{*,\lambda}}{P_s}dP_s}|\mathcal{F}_{t}]$.

The optimization problem (\ref{optm2}) is a special case of
(\ref{optm1}) with $\lambda=0$. Finally, by Definition
\ref{definition1}, the price $\mathfrak{C}_t^{\lambda}$ is given by
the solution to
\begin{align*}
-e^{-\gamma\left(X_t-
\mathfrak{C}_t^{\lambda}+{Y}_t^{\lambda}\right)}e^{\gamma\int_0^t\theta_sds}&
=E^{\mathbf{P}}\left[\left.-e^{-\gamma\left(X_T^{X_t-\mathfrak{C}_t^{\lambda}}
({\pi}^{*,\lambda})+\lambda
g({\cal{S}}_{\cdot})\right)}\right|\mathcal{F}_t\right]\\[+0.1cm]
&=E^{\mathbf{P}}\left[\left.-e^{-\gamma
X_T^{X_t}(\pi^{*,0})}\right|\mathcal{F}_t\right]=-e^{-\gamma\left(X_t+{Y}_t^{0}\right)}e^{\gamma\int_0^t\theta_sds}.
\end{align*}
Therefore, $\mathfrak{C}_t^{\lambda}=Y_t^{\lambda}-Y_t^{0}$, and the
hedging strategy for $\lambda$ units of the option is given by
$${\pi}^{*,\lambda}_t-\pi^{*,0}_t=
-\frac{\langle\sigma_t^P,\mathcal{Z}^{\lambda}_t\rangle}{||\sigma_t^P||^2}+\frac{\mu_t^P}{\gamma||\sigma_t^P||^2}+\frac{\langle\sigma_t^P,\mathcal{Z}^{0}_t\rangle}{||\sigma_t^P||^2}-\frac{\mu_t^P}{\gamma||\sigma_t^P||^2}
=-\frac{\langle\sigma_t^P,\mathcal{Z}_t^{\lambda}-\mathcal{Z}_t^{0}\rangle}{||\sigma_t^P||^2},$$
which completes the proof.
\end{proof}

\section{Functional Differential Equations and Pseudo Linear Pricing Rule}\label{sec-pseudo}

In this section, we present our pseudo linear pricing rule for the
utility indifference price $\mathfrak{C}^{\lambda}_t$. The main idea
is motivated by Liang et al \cite{LLQ2009}, where the authors
introduce a class of functional differential equations in order to
solve BSDEs on a general filtered probability space. See also
Casserini and Liang \cite{Casserini} for a generalization of this
method to solve FBSDEs. The solution $Y$ to BSDE (\ref{BSDE for n
nontraded assets}) is a bounded special semimartingale, so admits a
unique decomposition under $\mathbf{P}$:
$$Y_t^{\lambda}=M_t^{\lambda,\mathbf{P}}-V_t^{\lambda,\mathbf{P}},\ \ \ \text{for}\ t\in[0,T],$$
where $M^{\lambda,\mathbf{P}}$ is the martingale part, which is a
$\mathbf{P}$-\emph{BMO} martingale, and $V^{\lambda,\mathbf{P}}$ is
the finite variation part with $V^{\lambda,\mathbf{P}}_0=0$. By the
martingale property of $M^{\lambda,\mathbf{P}}$,
\begin{align}\label{algebra}
Y_t^{\lambda}
&=E^{\mathbf{P}}\left[\left.Y_T^{\lambda}+V_T^{\lambda,\mathbf{P}}\right|\mathcal{F}_t\right]-V_t^{\lambda,\mathbf{P}}\nonumber\\[+0.1cm]
&=E^{\mathbf{P}}\left[\left.\lambda
g(\mathcal{S}_{\cdot})+\int_0^T\theta_sds\right|\mathcal{F}_t\right]
+E^{\mathbf{P}}[V_T^{\lambda,\mathbf{P}}-V_t^{\lambda,\mathbf{P}}|\mathcal{F}_t].
\end{align}
In other words, knowing the finite variation process
$V^{\lambda,\mathbf{P}}$ and the terminal data $Y_T^{\lambda}$ is
enough to calculate $Y^{\lambda}$, which in turn gives us a new
pricing rule for the utility indifference price
$\mathfrak{C}^{\lambda}_t$.

\begin{theorem}\label{lemmapesudo} (Pseudo linear pricing rule)

Suppose that Assumptions (A1) (A2), and (A3) are satisfied. If
$V^{\lambda,\mathbf{P}}$ is the unique solution of the following
functional differential equation
\begin{equation}\label{FDE}
V_t^{\lambda,\mathbf{P}}=\int_0^tF_s(\mathcal{Z}^{\lambda,\mathbf{P}}_s(V^{\lambda,\mathbf{P}}))ds,
\end{equation}
with $\mathcal{Z}^{\lambda,\mathbf{P}}(\cdot)$, as an affine
functional of a stochastic process $V$, given by
\begin{align*}
\int_t^{T}\langle\mathcal{Z}^{\lambda,\mathbf{P}}_s(V),d\mathcal{W}_s\rangle
=&\ \left(\lambda
g(\mathcal{S}_{\cdot})+\int_0^T\theta_sds\right)+V_{T}\\
&\ - E^{\mathbf{P}}\left[\left.\left(\lambda
g(\mathcal{S}_{\cdot})+\int_0^T\theta_sds\right)+V_{T}\right|\mathcal{F}_t\right],
\end{align*}
then the utility indifference price $\mathfrak{C}_t^{\lambda}$ can
be represented by the following linear conditional expectation
\begin{equation}\label{Linearprcingrule}
\mathfrak{C}_t^{\lambda}=E^{\mathbf{P}}\left[\left.\lambda
g(\mathcal{S}_{\cdot})\right|\mathcal{F}_t\right]+
E^{\mathbf{P}}[V_T^{\lambda,\mathbf{P}}-V_t^{\lambda,\mathbf{P}}|\mathcal{F}_t]-
E^{\mathbf{P}}[V_T^{0,\mathbf{P}}-V_t^{0,\mathbf{P}}|\mathcal{F}_t],
\end{equation}
and the hedging strategy for $\lambda$ units of the option is given
by
$$-\frac{\langle\sigma_t^P,\mathcal{Z}_t^{\lambda,\mathbf{P}}(V^{\lambda,\mathbf{P}})
-\mathcal{Z}_t^{0,\mathbf{P}}(V^{0,\mathbf{P}})\rangle}{||\sigma_t^P||^2}.$$
\end{theorem}

\begin{proof} To obtain the functional differential equation
(\ref{FDE}), we take conditional expectation of (\ref{BSDE for n
nontraded assets}) on $\mathcal{F}_t$:
\begin{align*}
Y_t^{\lambda}&=E^{\mathbf{P}}\left[\left.\left(\lambda
g({\cal{S}}_{\cdot})+\int_0^T\theta_sds\right)+\int_t^TF_s({\cal{Z}}^{\lambda}_s)ds\right|\mathcal{F}_t\right]\\
&=E^{\mathbf{P}}\left[\left.\left(\lambda
g({\cal{S}}_{\cdot})+\int_0^T\theta_sds\right)+\int_0^TF_s({\cal{Z}}^{\lambda}_s)ds\right|\mathcal{F}_t\right]
-\int_0^tF_s({\cal{Z}}^{\lambda}_s)ds.
\end{align*}
On the other hand, $Y^{\lambda}$ admits the decomposition
$Y_t^{\lambda}=M_t^{\lambda,\mathbf{P}}-V_{t}^{\lambda,\mathbf{P}}$.
Due to the uniqueness of special semimartingale decomposition, we
obtain (\ref{FDE}) by identifying the finite variation parts of the
above two expressions for $Y^{\lambda}$. To show that
$\mathcal{Z}_t^{\lambda}=\mathcal{Z}_t^{\lambda,\mathbf{P}}(V^{\lambda,\mathbf{P}})$,
we only need to note that $Z^{\lambda}$ is the martingale
representation of $M^{\lambda,\mathbf{P}}$. Finally,
(\ref{Linearprcingrule}) follows from (\ref{representationformula1})
and (\ref{algebra}).
\end{proof}

Since (\ref{Linearprcingrule}) is under linear expectation, we call
it the \emph{pseudo linear pricing rule} for the utility
indifference price $\mathfrak{C}_t^{\lambda}$. The advantage of this
pricing rule compared to the nonlinear pricing rule
(\ref{representationformula1}) is that we only need to solve
functional differential equation (\ref{FDE}) for
$V^{\lambda,\mathbf{P}}$ in order to calculate the utility
indifference price $\mathfrak{C}_t^{\lambda}$. Moreover, functional
differential equation (\ref{FDE}) runs forwards. Thus, the
conflicting nature between the backward equation and the underlying
forward equation is avoided.

Whereas in \cite{LLQ2009} the driver is Lipschitz continuous, the
driver $F_t(z)$ of functional differential equation (\ref{FDE}) is
quadratic in $z$. Nevertheless, we can employ Picard iteration to
approximate $V^{\lambda,\mathbf{P}}$, and therefore the utility
indifference price $\mathfrak{C}_t^{\lambda}$. We rely mainly on the
change of probability measure, and we will present two linear
approximations for $\mathfrak{C}_t^{\lambda}$ depending on different
choices of probability measures. We first present an equivalent
formulation of pseudo linear pricing rule but under a different
probability measure.

\begin{corollary} \label{changemeasurecorollary}
Let ${\cal{B}}=(B^1,\cdots,B^d)$ be a $d$-dimensional Brownian
motion on a filtered probability space
$(\Omega,\mathcal{F},\{\mathcal{F}_t\},\mathbf{Q})$.  Suppose that
$N$ is some $\mathbf{Q}$-\text{BMO} martingale, so that its
Dol\'eans-Dade exponential $\mathcal{E}(N)$ is uniformly integrable.
Define
$$\frac{d\mathbf{P}}{d\mathbf{Q}}=\mathcal{E}(N).$$

Suppose that $V^{\lambda,\mathbf{Q}}$ solves the following
functional differential equation
\begin{equation}\label{FDE_change_of_measure}
V_t^{\lambda,\mathbf{Q}}=\int_0^tF_s(\mathcal{Z}_s^{\lambda,\mathbf{Q}}(V^{\lambda,\mathbf{Q}}))ds+
\langle\mathcal{Z}_s^{\lambda,\mathbf{Q}}(V^{\lambda,\mathbf{Q}}),d[\mathcal{B},N]_s\rangle,
\end{equation}
with $\mathcal{Z}^{\lambda,\mathbf{Q}}(\cdot)$, as an affine
functional of $V$, given by
\begin{align*}
\int_t^{T}\langle\mathcal{Z}_s^{\lambda,\mathbf{Q}}(V),d\mathcal{B}_s\rangle
=&\ \left(\lambda
g(\mathcal{S}_{\cdot})+\int_0^T\theta_sds\right)+V_{T}\\
&\ -E^{\mathbf{Q}}\left[\left.\left(\lambda
g(\mathcal{S}_{\cdot})+\int_0^T\theta_sds\right)+V_{T}\right|\mathcal{F}_t\right],
\end{align*}
and $\mathcal{S}=(S^1,\cdots,S^d)$ given by
\begin{equation}\label{dynamics of nontraded assets_change_meaure}
S_t^i=S_0^i+\int_0^tS_s^i(\mu^i_sds- \langle
\sigma^i_s,d[\mathcal{B},N]_s\rangle+ \langle
\sigma^i_s,d\mathcal{B}_s\rangle).
\end{equation}
Then
\begin{equation}\label{formulaforVP}
V_t^{\lambda,\mathbf{P}}=V_t^{\lambda,\mathbf{Q}}-
\int_0^t\langle\mathcal{Z}_s^{\lambda,\mathbf{Q}}(V^{\lambda,\mathbf{Q}}),d[\mathcal{B},N]_s\rangle
\end{equation}
solves (\ref{FDE}) on the filtered probability space
$(\Omega,\mathcal{F},\{\mathcal{F}_t\},\mathbf{P})$.
\end{corollary}

\begin{proof} By the
definition of $V^{\lambda,\mathbf{P}}$ in (\ref{formulaforVP}) and
functional differential equation (\ref{FDE_change_of_measure}), we
have
$$V_t^{\lambda,\mathbf{P}}=\int_0^tF_s(\mathcal{Z}_s^{\lambda,\mathbf{Q}}(V^{\lambda,\mathbf{Q}}))ds.$$
Hence, we only need to show that
$\mathcal{Z}_t^{\lambda,\mathbf{Q}}(V^{\lambda,\mathbf{Q}})=\mathcal{Z}_t^{\lambda,\mathbf{P}}(V^{\lambda,\mathbf{P}})$
for $t\in[0,T]$, which means
$\mathcal{Z}^{\lambda,\mathbf{Q}}(V^{\lambda,\mathbf{Q}})$ is
invariant under the change of probability measure. In other words,
the martingale representation is invariant under the change of
probability measure:
$$\int_t^{T}\langle\mathcal{Z}^{\lambda,\mathbf{Q}}_s(V^{\lambda,\mathbf{Q}}),d\mathcal{W}_s\rangle=M_T^{\lambda,\mathbf{P}}-M_t^{\lambda,\mathbf{P}},$$
with
$$M_t^{\lambda,\mathbf{P}}=E^{\mathbf{P}}\left[\left.\left(\lambda
g(\mathcal{S}_{\cdot})+\int_0^T\theta_sds\right)+V_{T}^{\lambda,\mathbf{P}}\right|\mathcal{F}_t\right],$$
and $\mathcal{W}=\mathcal{B}-[\mathcal{B},N]$ being Brownian motion
under the probability measure $\mathbf{P}$. Indeed, using Girsanov's
transformation,
\begin{align*}
&\
\int_t^{T}\langle\mathcal{Z}^{\lambda,\mathbf{Q}}_s(V^{\lambda,\mathbf{Q}}),d\mathcal{W}_s\rangle\\
=&\
\int_t^{T}\langle\mathcal{Z}^{\lambda,\mathbf{Q}}_s(V^{\lambda,\mathbf{Q}}),d\mathcal{B}_s\rangle,
-\int_t^{T}\langle\mathcal{Z}^{\lambda,\mathbf{Q}}_s(V^{\lambda,\mathbf{Q}}),d[\mathcal{B},N]_s\rangle\\
=&\ \left(\lambda
g(\mathcal{S}_{\cdot})+\int_0^T\theta_sds\right)+V_{T}^{\lambda,\mathbf{Q}}
-E^{\mathbf{Q}}\left[\left.\left(\lambda
g(\mathcal{S}_{\cdot})+\int_0^T\theta_sds\right)+V_{T}^{\mathbf{\lambda,\mathbf{Q}}}\right|\mathcal{F}_t\right]\\
&\ -
\int_t^{T}\langle\mathcal{Z}^{\lambda,\mathbf{Q}}_s(V^{\lambda,\mathbf{Q}}),d[\mathcal{B},N]_s\rangle
\\
=&\ \left(\lambda
g(\mathcal{S}_{\cdot})+\int_0^T\theta_sds\right)+V_{T}^{\lambda,\mathbf{P}}
-E^{\mathbf{P}}\left[\left.\left(\lambda
g(\mathcal{S}_{\cdot})+\int_0^T\theta_sds\right)+V_{T}^{\mathbf{\lambda,\mathbf{P}}}\right|\mathcal{F}_t\right]\\[+0.2cm]
=&\ M_T^{\lambda,\mathbf{P}}-M_t^{\lambda,\mathbf{P}}.
\end{align*}
\end{proof}

\subsection{Perturbations of Functional Differential Equations
}

By the pseudo linear pricing rule (\ref{Linearprcingrule}), we only
need to solve functional differential equation (\ref{FDE}) for
$V^{\lambda,\mathbf{P}}$ in order to obtain the utility indifference
price $\mathfrak{C}^{\lambda}_t$. Our first linear approximation for
the utility indifference price $\mathfrak{C}^{\lambda}_t$ is based
on perturbations of functional differential equation (\ref{FDE}),
the idea of which is motivated by Proposition 2 of Tevzadze
\cite{Tevzadze}\footnote{We thank one of the referees for the
suggestion of this method.}.

We first decompose the units of the option $\lambda$ as the
following finite sum: $\lambda=\sum_{j=1}^{J}\lambda_{j}$ such that
$$\lambda_j\leq\frac{\lambda}{32K_1K_2^2},$$
where $K_1$ is the constant from the John-Nirenburg inequality
(\ref{JN_inequality}), and $K_2$ is the constant from BSDE
(\ref{BSDE for n nontraded assets}) (see (\ref{const0}) and
(\ref{const1})). We then make perturbations of functional
differential equation (\ref{FDE}) as follows:
\begin{equation}\label{FDEpertubation}
V_t^{\lambda_j,\mathbf{P}}=\int_0^t
F_s\left(\sum_{k=1}^{j}\mathcal{Z}_s^{\lambda_k,\mathbf{P}}(V^{\lambda_k,\mathbf{P}})\right)-
F_s\left(\sum_{k=1}^{j-1}\mathcal{Z}_s^{\lambda_k,\mathbf{P}}(V^{\lambda_k,\mathbf{P}})\right)ds,
\end{equation}
with $\mathcal{Z}^{\lambda_{j},\mathbf{P}}(\cdot)$, as an affine
functional of $V$, given by
\begin{align*}
\int_t^{T}\langle\mathcal{Z}_s^{\lambda_j,\mathbf{P}}(V),d\mathcal{W}_s\rangle
=&\ \frac{\lambda_j}{\lambda}\left(\lambda
g(\mathcal{S}_{\cdot})+\int_0^T\theta_sds\right)+V_{T}\\
&\ -
E^{\mathbf{P}}\left[\left.\frac{\lambda_j}{\lambda}\left(\lambda
g(\mathcal{S}_{\cdot})+\int_0^T\theta_sds\right)+V_{T}\right|\mathcal{F}_t\right],
\end{align*}
and with $\sum_{k=1}^{0}=0$ by convention. Then it is easy to verify
that
$V_t^{\lambda,\mathbf{P}}=\sum_{j=1}^{J}V_t^{\lambda_j,\mathbf{P}}$
solves functional differential equation (\ref{FDE}). For functional
differential equation (\ref{FDEpertubation}), we can give the
following linear approximation for its solution
$V^{\lambda_{j},\mathbf{P}}$.\\

Define the Banach space $\mathfrak{V}([0,T];\mathbb{R})$ for the
continuous and $\mathcal{F}_t$-adapted processes valued in
$\mathbb{R}$, endowed with the norm
$$||V||_{\mathfrak{V}[0,T]}=\sup_{\tau}\left\|E[|V_T-V_{\tau}||\mathcal{F}_{\tau}]\right\|_{\infty}$$
for any $\mathcal{F}_t$-stopping time $\tau\in[0,T]$. Furthermore,
define its subspace
$$\mathfrak{V}([0,T];B_r)=\left\{V\in\mathfrak{V}([0,T];\mathbb{R}):
||V||_{\mathfrak{V}[0,T]}\leq r\ \text{for}\ r=
\frac{1}{32K_1K_2}\right\}.$$

\begin{proposition}\label{proposition1} Let ${\cal{B}}=(B^1,\cdots,B^d)$ be
a $d$-dimensional Brownian motion on a filtered probability space
$(\Omega,\mathcal{F},\{\mathcal{F}_t\},\mathbf{Q})$. For fixed $j$
where $1\leq j\leq J$, suppose that we have solved functional
differential equation (\ref{FDEpertubation}) and obtained its
solution $V^{\lambda_k,\mathbf{P}}$ for $k=1,\cdots,j-1$, so that we
have the affine functionals
$\mathcal{Z}^{\lambda_k,\mathbf{P}}(\cdot)$. Then $N^j$ defined by
$$N^{j}=\int_0^{\cdot}\langle\nabla_zF_s\left(\sum_{k=1}^{j-1}
\mathcal{Z}_s^{\lambda_k,\mathbf{P}}(V^{\lambda_k,\mathbf{P}})\right),d\mathcal{B}_s\rangle$$
is a $\mathbf{Q}$-\text{BMO} martingale.

Define the following sequence
$\{V^{\lambda_j,\mathbf{Q}}(m)\}_{m\geq 0}$ iteratively:
$V^{\lambda_j,\mathbf{Q}}(0)=0$,
\begin{equation*}
V_t^{\lambda_j,\mathbf{Q}}(m+1)=
\int_0^t\tilde{F}_s\left(\mathcal{Z}_s^{\lambda_j,\mathbf{Q}}(V^{\lambda_j,\mathbf{Q}}(m))\right)ds,
\end{equation*}
with $\tilde{F}_s\left(z\right)$ given by
\begin{align*}
\tilde{F}_s^{j}(z)=&\
F_s\left(\sum_{k=1}^{j-1}\mathcal{Z}_s^{\lambda_k,\mathbf{P}}(V^{\lambda_k,\mathbf{P}})+z\right)-
F_s\left(\sum_{k=1}^{j-1}\mathcal{Z}_s^{\lambda_k,\mathbf{P}}(V^{\lambda_k,\mathbf{P}})\right)\\
&\ -\langle\nabla_zF_s\left(\sum_{k=1}^{j-1}
\mathcal{Z}_s^{\lambda_k,\mathbf{P}}(V^{\lambda_k,\mathbf{P}})\right),z\rangle.
\end{align*}
Then $\{V^{\lambda_{j},\mathbf{Q}}(m)\}_{m\geq 0}$ converges to some
$V^{\lambda_{j},\mathbf{Q}}$ in the space $\mathfrak{V}([0,T];B_r)$
with the convergence rate
$$||V^{\lambda_{j},\mathbf{Q}}-V^{\lambda_{j},\mathbf{Q}}(m)||_{\mathfrak{V}[0,T]}\leq r\left(\frac{1}{2}\right)^{m-1},$$
and $V^{\lambda_j,\mathbf{P}}$ is obtained by (\ref{formulaforVP}):
$$V_t^{\lambda_j,\mathbf{P}}=V_t^{\lambda_j,\mathbf{Q}}-
\int_0^t\langle
\mathcal{Z}_s^{\lambda_j,\mathbf{Q}}(V^{\lambda_j,\mathbf{Q}}),
\nabla_zF_s\left(\sum_{k=1}^{j-1}
\mathcal{Z}_s^{\lambda_k,\mathbf{P}}(V^{\lambda_k,\mathbf{P}})\right)\rangle
ds.$$
\end{proposition}

\begin{proof} For fixed $j$, we first verify that $N^{j}$ is a
$\mathbf{Q}$-\emph{BMO} martingale. By Corollary
\ref{changemeasurecorollary}, we have
$$\mathcal{Z}_s^{\lambda_k,\mathbf{P}}(V^{\lambda_k,\mathbf{P}})=\mathcal{Z}_s^{\lambda_k,\mathbf{Q}}(V^{\lambda_k,\mathbf{Q}})$$
for $k=1,\cdots,j-1$. Therefore,
\begin{align*}
&\ \sup_{\tau}\left\|E^{\mathbf{Q}}[|N_T^j-N_{\tau}^j|^2|\mathcal{F}_{\tau}]\right\|_{\infty}\\
=&\
\sup_{\tau}\left\|E^{\mathbf{Q}}\left[\left.\int_t^T\left|\nabla_{z}F_s\left(\sum_{k=1}^{j-1}
\mathcal{Z}_s^{\lambda_k,\mathbf{Q}}(V^{\lambda_k,\mathbf{Q}})\right)\right|^2ds\right|\mathcal{F}_{\tau}\right]\right\|_{\infty}\\
\leq&\
\sup_{\tau}\left\|E^{\mathbf{Q}}\left[\left.\int_t^T2K_2^2(1+||\mathcal{Z}_s^{\lambda_k,\mathbf{Q}}(V^{\lambda_k,\mathbf{Q}})||^2)ds
\right|\mathcal{F}_{\tau}\right]\right\|_{\infty}\\
<&\ \infty,
\end{align*}
where we used (\ref{const0}) and the $\mathbf{Q}$-\emph{BMO}
martingale property of
$$\int_0^{\cdot}\langle\mathcal{Z}_s^{\lambda_k,\mathbf{Q}}(V^{\lambda_k,\mathbf{Q}}),d\mathcal{B}_s\rangle.$$

Next, we consider the convergence of the sequence
$\{V^{\lambda_j,\mathbf{Q}}(m)\}_{m\geq 0}$. Similar to Remark 1 of
\cite{Tevzadze}, by using the mean value theorem twice on
$\tilde{F}_s^{j}(\cdot)$ and by using (\ref{const0}), it is easy to
verify that
\begin{align*}
|\tilde{F}_s^{j}(z)-\tilde{F}_s^{j}(\bar{z})|\leq K_2(||z||+||\bar{z}||)||z-\bar{z}||,\\
|\tilde{F}_s^{j}(z)|\leq K_2||z||^2.
\end{align*}
Given that $V^{\lambda_j,\mathbf{Q}}(m)\in\mathfrak{V}([0,T];B_r)$,
we need to verify that $V^{\lambda_j,\mathbf{Q}}(m+1)$ is in the
same space $\mathfrak{V}([0,T];B_r)$. Indeed,
\begin{align}\label{inequality1}
&\ ||V^{\lambda_j,\mathbf{Q}}(m+1)||_{\mathfrak{V}[0,T]}\nonumber\\
=&\ \sup_{\tau}\left\|E^{\mathbf{Q}}
\left[\left.\left|\int_{\tau}^T\tilde{F}_s\left(\mathcal{Z}_s^{\lambda_j,\mathbf{Q}}(V^{\lambda_j,\mathbf{Q}}(m))\right)ds\right|\ \right|\mathcal{F}_{\tau}\right]\right\|_{\infty}\notag\\
\leq&\ K_2\sup_{\tau}\left\|E^{\mathbf{Q}}
\left[\left.\int_{\tau}^T\left\|\mathcal{Z}_s^{\lambda_j,\mathbf{Q}}(V^{\lambda_j,\mathbf{Q}}(m))\right\|^2ds\right|\mathcal{F}_{\tau}\right]\right\|_{\infty}\notag\\
=&\ K_2\sup_{\tau}\left\|E^{\mathbf{Q}}
\left[\left.\left|\int_{\tau}^T\langle\mathcal{Z}_s^{\lambda_j,\mathbf{Q}}(V^{\lambda_j,\mathbf{Q}}(m)),d\mathcal{B}_s\rangle\right|^2\right|\mathcal{F}_{\tau}\right]\right\|_{\infty}\notag\\
\leq&\ K_1K_2\sup_{\tau}\left\|E^{\mathbf{Q}}
\left[\left.\left|\int_{\tau}^T\langle\mathcal{Z}_s^{\lambda_j,\mathbf{Q}}(V^{\lambda_j,\mathbf{Q}}(m)),d\mathcal{B}_s\rangle\right|\
\right|\mathcal{F}_{\tau}\right]\right\|_{\infty}^2,
\end{align}
where we used the John-Nirenburg inequality (\ref{JN_inequality}) in
the last inequality. In the following, we denote
$$\xi_{j}=\frac{\lambda_j}{\lambda}\left(\lambda
g(\mathcal{S}_{\cdot})+\int_0^T\theta_sds\right).$$ By the
definition of $\lambda_{j}$ and (\ref{const1}), $$|\xi_{j}|\leq
\frac{1}{32K_1K_2}.$$ With the notation $\xi_j$, the affine
functional $\mathcal{Z}^{\lambda_j,\mathbf{Q}}(\cdot)$ is rewritten
as
$$\int_t^{T}\langle\mathcal{Z}_s^{\lambda_j,\mathbf{Q}}(V),d\mathcal{B}_s\rangle
=\xi_j+V_T-E^{\mathbf{Q}}[\xi_j+V_T|\mathcal{F}_t].$$ Therefore,
following (\ref{inequality1}),
$||V^{\lambda_j,\mathbf{Q}}(m+1)||_{\mathfrak{V}[0,T]}$ is further
dominated by
\begin{align*}
&\ K_1K_2\sup_{\tau}\left\|E^{\mathbf{Q}}
\left[\left.\left|\xi_j+V_{T}^{\lambda_j,\mathbf{Q}}(m)-E^{\mathbf{Q}}[\xi_j+V_{T}^{\lambda_j,\mathbf{Q}}(m)|\mathcal{F}_{\tau}]\right|\ \right|\mathcal{F}_{\tau}\right]\right\|_{\infty}^2\\
=&\ K_1K_2\sup_{\tau}\left\|E^{\mathbf{Q}}
\left[\left|\xi_j-E^{\mathbf{Q}}[\xi_j|\mathcal{F}_{\tau}]+V_{T}^{\lambda_j,\mathbf{Q}}(m)-V_{\tau}^{\lambda_j,\mathbf{Q}}(m)\right.\right.\right.\\
&\ \ \ \ \ \ \ \ \ \ \ \ \ \ \
\left.\left.\left.\left.-E^{\mathbf{Q}}[V_{T}^{\lambda_j,\mathbf{Q}}(m)-V_{\tau}^{\lambda_j,\mathbf{Q}}(m)|\mathcal{F}_{\tau}]\right|\ \right|\mathcal{F}_{\tau}\right]\right\|_{\infty}^2\\
\leq&\
4K_1K_2\left(\sup_{\tau}\left\|E^{\mathbf{Q}}\left[|\xi_{j}||\mathcal{F}_{\tau}\right]\right\|^2_{\infty}
+\sup_{\tau}\left\|E^{\mathbf{Q}}\left[\left|E^{\mathbf{Q}}\left(\xi_{j}|\mathcal{F}_{\tau}\right)\right||\mathcal{F}_{\tau}\right]\right\|^2_{\infty}\right.\\
&\ \ \ \ \ \ \ \ \ \ \ \
+\sup_{\tau}\left\|E^{\mathbf{Q}}\left[\left.\left|V_T^{\lambda_{j},\mathbf{Q}}(m)-V_{\tau}^{\lambda_{j},\mathbf{Q}}(m)\right|\ \right|\mathcal{F}_{\tau}\right]\right\|_{\infty}^{2}\\
&\ \ \ \ \ \ \ \ \ \ \ \
\left.+\sup_{\tau}\left\|E^{\mathbf{Q}}\left[\left.\left|E^{\mathbf{Q}}\left(V_T^{\lambda_{j},\mathbf{Q}}(m)-V_{\tau}^{\lambda_{j},\mathbf{Q}}(m)|\mathcal{F}_{\tau}\right)\right|\ \right|\mathcal{F}_{\tau}\right]\right\|_{\infty}^{2}\right)\\
\leq&\
8K_1K_2\left(|\xi_j|^2+||V^{\lambda_{j},\mathbf{Q}}(m)||^2_{\mathfrak{V}[0,T]}\right)\leq
1/(64K_1K_2)\leq r.
\end{align*}

Similarly, we consider the difference $\delta
V^{\lambda_j,\mathbf{Q}}(m)=V^{\lambda_j,\mathbf{Q}}(m+1)-V^{\lambda_j,\mathbf{Q}}(m)$,
\begin{align*}
&\ ||\delta V^{\lambda_j,\mathbf{Q}}(m)||_{\mathfrak{V}[0,T]}^2\\
\leq&\ 2K_1^2K_2^2\sup_{\tau}\left\|E^{\mathbf{Q}}
\left[\left.\left|\int_t^T\langle\delta\mathcal{Z}^{\lambda_j,\mathbf{Q}}(V^{\lambda_j,\mathbf{Q}}(m-1)),d\mathcal{B}_s\rangle\right|\ \right|\mathcal{F}_{\tau}\right]\right\|^2_{\infty}\\
&\ \times\left\{ \sup_{\tau}\left\|E^{\mathbf{Q}}
\left[\left.\left|\int_t^T\langle\mathcal{Z}^{\lambda_j,\mathbf{Q}}(V^{\lambda_j,\mathbf{Q}}(m)),d\mathcal{B}_s\rangle\right|\ \right|\mathcal{F}_{\tau}\right]\right\|^2_{\infty}\right.\\
&\ \ \ \ \ \ \left.+\sup_{\tau}\left\|E^{\mathbf{Q}}
\left[\left.\left|\int_t^T\langle\mathcal{Z}^{\lambda_j,\mathbf{Q}}(V^{\lambda_j,\mathbf{Q}}(m-1)),d\mathcal{B}_s\rangle\right|\
\right|\mathcal{F}_{\tau}\right]\right\|^2_{\infty}
\right\}\\
\leq&\ 2K_1^2K_2^2\times4||\delta
V^{\lambda_{j},\mathbf{Q}}(m-1)||^2_{\mathfrak{V}[0,T]}\\
&\ \times2\left(8||\xi_j||^2+
4||V^{\lambda_{j},\mathbf{Q}}(m)||^2_{\mathfrak{V}[0,T]}+
4||V^{\lambda_{j},\mathbf{Q}}(m-1)||^2_{\mathfrak{V}[0,T]}\right)\\
\leq&\ \frac14||\delta
V^{\lambda_{j},\mathbf{Q}}(m-1)||^2_{\mathfrak{V}[0,T]}.
\end{align*}
We iterate the above inequality, and obtain
$$||\delta
V^{\lambda_{j},\mathbf{Q}}(m)||_{\mathfrak{V}[0,T]}\leq\left(\frac{1}{2}\right)^m||
V^{\lambda_{j},\mathbf{Q}}(1)||_{\mathfrak{V}[0,T]}\leq
\left(\frac{1}{2}\right)^m\frac{1}{32K_1K_2}.$$ Hence, for any
natural number $p$,
\begin{align*}
&\
||V^{\lambda_j,\mathbf{Q}}(m+p)-V^{\lambda_j,\mathbf{Q}}(m)||_{\mathfrak{V}[0,T]}\\
\leq&\
\sum_{j=1}^{p}||V^{\lambda_j,\mathbf{Q}}(m+j)-V^{\lambda_j,\mathbf{Q}}(m+j-1)||_{\mathfrak{V}[0,T]}\\
\leq&\ \left(\frac12\right)^{m-1}\times\frac{1}{32K_1K_2}.
\end{align*}
Letting $m\uparrow\infty$, we deduce that
$\{V^{\lambda_j,\mathbf{Q}}(m)\}_{m\geq 0}$ is a Cauchy sequence,
and converges to some limit $V^{\lambda_j,\mathbf{Q}}$. On the other
hand, letting $p\uparrow\infty$, we obtain the convergence rate.
\end{proof}

\subsection{Nonlinear Girsanov's Transformation}

The crucial step for our pseudo linear pricing rule
(\ref{Linearprcingrule}) is to solve functional differential
equation (\ref{FDE}) in order to obtain $V^{\lambda,\mathbf{P}}$.
Our second linear approximation for the utility indifference price
$\mathfrak{C}^{\lambda}_t$ is based on a nonlinear version of
Girsanov's transformation in order to vanish the driver $F_t(z)$ of
(\ref{FDE}). Such an idea has been known in the BSDE literature. For
example, it appears as Proposition 11 in Mania and Schweizer
\cite{MS} and measure solutions of BSDEs in Ankirchner et al
\cite{Imkeller0}, where they model the terminal data (the payoff) as
a general random variable. In contrast, as we specify the dynamics
of the underlying and the payoff structure, a coupled FBSDE appears
naturally.

The intuitive idea is as follows: Note that in Corollary
\ref{changemeasurecorollary}, if we choose $N$:
\begin{align}\label{nonlinear_measure}
N=&\ \int_0^{\cdot}\frac{\gamma}{2}\langle\mathcal{Z}^{\lambda,\mathbf{Q}}_s(V^{\lambda,\mathbf{Q}}),d\mathcal{B}_s\rangle\nonumber\\
&\
-\int_0^{\cdot}\frac{\gamma}{2||\sigma_s^P||^2}\left\{\langle\sigma_s^P,\mathcal{Z}^{\lambda,\mathbf{Q}}_s(V^{\lambda,\mathbf{Q}})\rangle-\frac{2\mu_s^P}{\gamma}\right\}\langle\sigma_s^P,d\mathcal{B}_s\rangle,
\end{align}
then the driver of functional differential equation
(\ref{FDE_change_of_measure}) will vanish, and
$V^{\lambda,\mathbf{Q}}_t=0$ for $t\in[0,T]$. It seems that
$V^{\lambda,\mathbf{P}}$ can be easily obtained by
(\ref{formulaforVP}) and (\ref{FDE_change_of_measure}):
$$V_t^{\lambda,\mathbf{P}}=0-
\int_0^t\langle\mathcal{Z}_s^{\lambda,\mathbf{Q}}(0),d[\mathcal{B},N]_s\rangle=0+
\int_0^tF_s(\mathcal{Z}_s^{\lambda,\mathbf{Q}}(0))ds.$$ However, in
this situation, the stochastic processes
$\mathcal{Z}^{\lambda,\mathbf{Q}}(0)$ and $\mathcal{S}$ depend on
each other as a loop:
\begin{equation}\label{backwardsde}
\int_t^{T}\langle\mathcal{Z}_s^{\lambda,\mathbf{Q}}(0),d\mathcal{B}_s\rangle
=\left(\lambda g(\mathcal{S}_{\cdot})+\int_0^T\theta_sds\right)
-E^{\mathbf{Q}}\left[\left.\left(\lambda
g(\mathcal{S}_{\cdot})+\int_0^T\theta_sds\right)\right|\mathcal{F}_t\right],
\end{equation}
and $\mathcal{S}=(S^1,\cdots,S^d)$ given by
\begin{equation}\label{forwardsde}
S_t^i=S_0^i+\int_0^tS_s^i(\mu^i_sds- \langle
\sigma^i_s,d[\mathcal{B},N]_s\rangle+ \langle
\sigma^i_s,d\mathcal{B}_s\rangle).
\end{equation}
Hence, we must solve (\ref{forwardsde}) as a functional differential
equation, which depends on $\mathcal{Z}^{\lambda,\mathbf{Q}}(0)$ as
a functional of the whole path of $\mathcal{S}$. Note that
(\ref{backwardsde}) and (\ref{forwardsde}) also form a special case
of coupled FBSDEs, if we introduce a backward process $Y^{\lambda}$
as conditional expectation:
$$Y_t^{\lambda}=E^{\mathbf{Q}}\left[\left.\left(\lambda
g(\mathcal{S}_{\cdot})+\int_0^T\theta_sds\right)\right|\mathcal{F}_t\right].$$

With the help of such a nonlinear Girsanov's transformation, the
remaining task is to solve functional differential equation
(\ref{forwardsde}). In the following, we work out its solution in a
special Markovian setting:
\begin{itemize}
\item \textbf{Assumption (A4)}: All the coefficients are
deterministic, and the payoff $g(\cdot)$ is a positive bounded
function satisfying:
$$|g(\mathcal{S}_T)-g(\bar{\mathcal{S}}_T)|\leq \frac{K_3}{\lambda}
\sum_{i=1}^d|\ln S_T^i-\ln\bar{S}_T^i|.$$
\end{itemize}

A typical example that we keep in mind is $g(s)=\min(K,s)$ for some
constant $K>0$. Under Assumptions (A1)-(A4), the conditional
expectation $Y_t^{\lambda}$ can be written as a function of time $t$
and the price process $\mathcal{S}_t$: $Y^{\lambda}_t=
Y^{\lambda}(t,\mathcal{S}_t)$. If we define the log process
$\mathcal{X}_t^i=\ln\mathcal{S}^i_t$ for $i=1,\cdots,d$, then
$Y^{\lambda}(t,e^{\mathcal{X}_t})$ is uniformly Lipschitz continuous
in $\mathcal{X}_t=(\mathcal{X}_t^1,\cdots,\mathcal{X}_t^d)$ (see
Theorem 2.9 of Delarue \cite{MR2053051}). We denote such a Lipschitz
constant still as $K_3$.

We first make a partition of $[0,T]$ as follows:
$0=t_0<t_1<\cdots<t_J=T$ such that
$$\max_{1\leq j\leq J}\Delta_j=\max_{1\leq j\leq J}|t_j-t_{j-1}|\leq \frac{1}{8K_3^2K_4},$$
where the definition of $K_4$ is given in the proof of the following
Proposition \ref{proposition2}. On each interval $[t_{j-1},t_{j}]$,
functional differential equation (\ref{forwardsde}) can be
reformulated in terms of $\mathcal{X}$:
\begin{equation}\label{FDE_for_x}
\mathcal{X}_t^{i}=\mathcal{X}_{t_{j-1}}^{i}+
\int_{t_{j-1}}^t(\mu^i_s-\frac12||\sigma_s^i||^2)ds- \langle
\sigma^i_s,d[\mathcal{B},N]_s\rangle+
\langle\sigma^i_s,d\mathcal{B}_s\rangle.
\end{equation}
For functional differential equation (\ref{FDE_for_x}), we can give
the following linear approximation for its solution $\mathcal{X}$.\\

Define the Banach space $\mathfrak{S}([t_{j-1},t_j];\mathbb{R}^d)$
for the continuous and $\mathcal{F}_t$-adapted processes valued in
$\mathbb{R}^d$, endowed with the norm:
$$||\mathcal{X}||_{\mathfrak{S}[t_{j-1},t_j]}=
E\left\{\sup_{t\in[t_{j-1},t_j]}||\mathcal{X}_t||^2\right\}^{1/2}.$$

\begin{proposition}\label{proposition2}
Let ${\cal{B}}=(B^1,\cdots,B^d)$ be a $d$-dimensional Brownian
motion on a filtered probability space
$(\Omega,\mathcal{F},\{\mathcal{F}_t\},\mathbf{Q})$. For fixed $j$
where $1\leq j\leq J$, suppose that we have solved functional
differential equation (\ref{FDE_for_x}) and obtained its solution
$\mathcal{X}_t$ for $t\in[t_j,t_{j+1}],\cdots,[t_{J-1},T]$, so that
we have a Lipschitz continuous function $Y^{\lambda}(t_j,e^{x})$ at
time $t_j$.

Define the sequence $\{\mathcal{X}(m)\}_{m\geq 0}$ on
$[t_{j-1},t_j]$ iteratively: $\mathcal{X}(0)=x$,
\begin{equation*}
\mathcal{X}_t^{i}(m+1)=x^{i}+
\int_{t_{j-1}}^t(\mu^i_s-\frac12||\sigma_s^i||^2)ds- \langle
\sigma^i_s,d[\mathcal{B},N(m)]_s\rangle+
\langle\sigma^i_s,d\mathcal{B}_s\rangle,
\end{equation*}
where $N(m)$ is given by (\ref{nonlinear_measure}) with
$\mathcal{Z}^{\lambda,\mathbf{Q}}(V^{\lambda,\mathbf{Q}})$ replaced
by $\mathcal{Z}^{\lambda,\mathbf{Q}}(0,m)$:
$$
\int_t^{t_j}\langle\mathcal{Z}_s^{\lambda,\mathbf{Q}}(0,m),d\mathcal{B}_s\rangle
=Y^{\lambda}\left(t_j,e^{\mathcal{X}_{t_j}(m)}\right)-
E^{\mathbf{Q}}\left[Y^{\lambda}\left(t_j,e^{\mathcal{X}_{t_j}(m)}\right)|\mathcal{F}_t\right].$$
Then $\{\mathcal{X}(m)\}_{m\geq 0}$ converges to some $\mathcal{X}$
in the space $\mathfrak{S}([t_{j-1},t_j];\mathbb{R}^d)$ with the
convergence rate
$$||\mathcal{X}-\mathcal{X}(m)||_{\mathfrak{S}[0,T]}\leq
\left(\frac12\right)^{m-1}\times||\mathcal{X}(1)-\mathcal{X}(0)||_{\mathfrak{S}[t_{j-1},t_j]},$$
from which we get a Lipschitz continuous function at time $t_{j-1}$
$$Y^{\lambda}(t_{j-1},e^{\mathcal{X}_{t_{j-1}}})=E^{\mathbf{Q}}\left[Y^{\lambda}\left(t_j,e^{\mathcal{X}_{t_j}}\right)|\mathcal{F}_{t_{j-1}}\right],$$
and $\mathcal{Z}^{\lambda,\mathbf{Q}}(0)$ on $[t_{j-1},t_j]$
$$\langle\int_{t_{j-1}}^{t_j}\mathcal{Z}^{\lambda,\mathbf{Q}}(0),d\mathcal{B}_s\rangle=
Y^{\lambda}\left(t_{j},e^{\mathcal{X}_{t_j}}\right)-Y^{\lambda}\left(t_{j-1},e^{\mathcal{X}_{t_{j-1}}}\right).$$
Finally, $N$ defined by (\ref{nonlinear_measure}) is a
$\mathbf{Q}$-\text{BMO} martingale, and $V^{\lambda,\mathbf{P}}$ is
obtained by (\ref{formulaforVP}) and (\ref{FDE_change_of_measure}):
$$V_t^{\lambda,\mathbf{P}}=
\int_0^tF_s(\mathcal{Z}_s^{\lambda,\mathbf{Q}}(0))ds.$$
\end{proposition}

\begin{proof}
The proof is similar to the proof of Proposition \ref{proposition1},
so we only sketch it. For fixed $j$, it is easy to verify that
$\mathcal{X}(m+1)\in\mathcal{S}[t_{j-1},t_j]$, if
$\mathcal{X}(m)\in\mathcal{S}[t_{j-1},t_j]$.

Next, consider the difference
$\delta\mathcal{X}(m)=\mathcal{X}(m+1)-\mathcal{X}(m)$,
\begin{align*}
&\ ||\delta\mathcal{X}(m)||^2_{\mathfrak{S}[t_{j-1},t_j]}\\
=&\
E^{\mathbf{Q}}\left\{\sup_{t\in[t_{j-1},t_j]}\sum_{i=1}^d\left|\int_{t_{j-1}}^{t}
\langle\sigma_s^i,d[\mathcal{B},\delta
N(m-1)]_s\rangle\right|^2\right\}\\
=&\ E^{\mathbf{Q}}\left\{\sup_{t}\sum_{i}
\left|\int_{t_{j-1}}^t\left(\frac{\gamma}{2}\langle\sigma_s^{i},\delta\mathcal{Z}_s^{\lambda,\mathbf{Q}}(0,m-1)\rangle\right.\right.\right.\\
&\ \ \ \ \ \
\left.\left.\left.-\frac{\gamma\langle\sigma_s^{i},\sigma_s^{P}\rangle}{2||\sigma_s^P||}\langle\sigma_s^P,\delta\mathcal{Z}_s^{\lambda,\mathbf{Q}}(0,m-1)\rangle
\right)ds\right|^2\right\}\\
\leq&\
K_4\Delta_jE^{\mathbf{Q}}\left\{\int_{t_{j-1}}^{t_j}||\delta\mathcal{Z}_s^{\lambda,\mathbf{Q}}(0,m-1)||^2ds\right\}\\
=&\
K_4\Delta_jE^{\mathbf{Q}}\left\{\left|\int_{t_{j-1}}^{t_j}\langle\delta\mathcal{Z}^{\lambda,\mathbf{Q}}(0,m-1),d\mathcal{B}_s\rangle\right|^2\right\}\\
\leq&\ 2K_4\Delta_j
E^{\mathbf{Q}}\left\{\left|Y^{\lambda}\left(t_{j},e^{\mathcal{X}_{t_j}(m)}\right)
-Y^{\lambda}\left(t_{j},e^{\mathcal{X}_{t_j}(m-1)}\right)\right|^2\right\}\\
\leq&\
2K_3^2K_4\Delta_{j}||\delta\mathcal{X}(m-1)||^2_{\mathfrak{S}[t_{j-1},t_j]}\leq\frac14
||\delta\mathcal{X}(m-1)||^2_{\mathfrak{S}[t_{j-1},t_j]}.
\end{align*}
We iterate the above inequality and obtain
$$||\delta\mathcal{X}(m)||_{\mathfrak{S}[t_{j-1},t_j]}\leq \left(\frac12\right)^{m}||\mathcal{X}(1)-\mathcal{X}(0)||_{\mathfrak{S}[t_{j-1},t_j]}.$$
Hence, for any natural number $p$,
\begin{align*}
||\mathcal{X}(m+p)-\mathcal{X}(m)||_{\mathfrak{S}[t_{j-1},t_j]}
&\leq
\sum_{j=1}^{p}||\mathcal{X}(m+j)-\mathcal{X}(m+j-1)||_{\mathfrak{S}[t_{j-1},t_j]}\\
&\leq\left(\frac12\right)^{m-1}\times||\mathcal{X}(1)-\mathcal{X}(0)||_{\mathfrak{S}[t_{j-1},t_j]}.
\end{align*}
Letting $m\uparrow\infty$, we deduce that $\{\mathcal{X}(m)\}_{m\geq
0}$ is a Cauchy sequence, and converges to some limit $\mathcal{X}$.
On the other hand, letting $p\uparrow\infty$, we obtain the
convergence rate.

The rest of the proof is to verify that $N$ defined by
(\ref{nonlinear_measure}) is a $\mathbf{Q}$-\emph{BMO} martingale,
which follows from the $\mathbf{Q}$-\emph{BMO} martingale property
of
$$\int_0^{\cdot}\langle\mathcal{Z}^{\lambda,\mathbf{Q}}_s(0),d\mathcal{B}_s\rangle.$$
\end{proof}

\end{document}